\PassOptionsToPackage{numbers,sort&compress}{natbib}

\documentclass{article}
\usepackage{iclr2026_conference}
\usepackage{amsmath,amssymb}
\usepackage{amsthm}
\usepackage[hidelinks]{hyperref} 

\newtheorem{definition}{Definition}
\newtheorem{lemma}{Lemma}
\newtheorem{theorem}{Theorem}

\title{AVEC: Bootstrapping Privacy for Local LLMs}

\author{
\begin{minipage}{0.9\linewidth}
\centering
		Madhava Gaikwad\\
		Microsoft \\
	\texttt{mgaikwad@microsoft.com} \\
\end{minipage}
}

\iclrfinalcopy
\begin{document}

\maketitle

\begin{abstract}
This position paper presents AVEC (Adaptive Verifiable Edge Control), a framework for bootstrapping privacy for local language models by enforcing privacy at the edge with explicit verifiability for delegated queries. AVEC introduces an adaptive budgeting algorithm that allocates per-query differential privacy parameters based on sensitivity, local confidence, and historical usage, and uses verifiable transformation with on-device integrity checks. We formalize guarantees using R\'enyi differential privacy with odometer-based accounting, and establish utility ceilings, delegation-leakage bounds, and impossibility results for deterministic gating and hash-only certification. Our evaluation is simulation-based by design to study mechanism behavior and accounting; we do not claim deployment readiness or task-level utility with live LLMs. The contribution is a conceptual architecture and theoretical foundation that chart a pathway for empirical follow-up on privately bootstrapping local LLMs.
\end{abstract}

\section{Introduction}

Large language models (LLMs) are increasingly integrated into decision making and user facing applications, yet their deployment is limited by persistent challenges in latency, privacy, and verifiable trust. On device LLMs improve privacy but remain constrained in accuracy and capability due to limited resources, while remote API based LLMs incur costs, delays, and risks of exposing sensitive data. Existing mitigations such as heuristic caching or fixed privacy policies provide only partial solutions, as they fail to adapt to query sensitivity and user context, and they offer no mechanism for independent verification of privacy enforcement. These limitations lead to the bootstrap problem in local LLM development, where privacy concerns and lack of initial training data hinder incremental improvement.

We propose AVEC (Adaptive Verifiable Edge Control), a framework that decentralizes privacy enforcement and introduces explicit verifiability for LLM interactions. AVEC moves budget allocation and verification to a local agent on the user device, minimizing delegation to remote services. The adaptive edge privacy budgeting algorithm computes a differential privacy budget per query by considering sensitivity, user preference, and historical consumption. A translation agent applies budget calibrated transformations and issues a proof of transformation that the local agent can verify before any data leaves the device. This mechanism ensures that privacy enforcement is not only applied but also auditable.

The contributions of this paper are threefold. First, we formalize adaptive edge privacy budgeting as a general method for allocating per query differential privacy guarantees and we provide an instantiation using entity level randomized response. Second, we introduce a verifiable transformation layer that binds declared parameters to observable outputs and prove impossibility of deterministic delegation and hash only proofs. Third, we present a formal privacy accounting model based on Rényi differential privacy and privacy odometers and prove bounds on utility ceilings and leakage through delegation. Together these components establish AVEC as a foundation for verifiable edge privacy in LLM interactions. We validate the framework in simulation across 30 trials and 25 domains, showing reduced delegation, lower cost and latency, and transparent privacy utility trade offs.

\section{Related Work}

Privacy in language model deployment has been studied mainly at training time and with static inference policies. Differential privacy for model training, most prominently through DP-SGD \cite{abadi2016deep}, provides strong guarantees but does not address inference-time privacy. On-device and edge-based models mitigate exposure \cite{alizadeh2024llmflash} but remain limited in capacity and adaptability. Federated learning offers collaborative training without raw data sharing \cite{mcmahan2017communication}, though it focuses on model updates rather than per-query privacy. Verifiable computation methods such as SNARKs \cite{parno2016pinocchio} and Bulletproofs \cite{boneh2018bulletproofs} can certify correctness of computations, but their costs remain prohibitive for interactive LLM inference. Semantic caching improves efficiency by reusing previous responses \cite{bang2023gptcache}, yet provides no formal privacy guarantee. Rényi differential privacy \cite{mironov2019rdp} and privacy odometers and filters \cite{rogers2016privacy} provide refined accounting for adaptive DP mechanisms, but have not been applied to inference-time LLM pipelines. AVEC connects these threads by combining adaptive edge privacy budgeting, verifiable proof of transformation, and formal DP accounting for query-level LLM interactions.
\section{Methodology}

AVEC is designed as a three tier architecture consisting of a local agent on the user device, a translation agent that applies privacy transformations, and a remote agent that executes standard LLM inference. The design goal is to move privacy enforcement and verification to the edge, ensuring that sensitive queries are transformed and budgeted before leaving the device. The methodology combines adaptive budgeting, differentially private transformations, and verifiable proofs of transformation into a coherent pipeline.

\subsection{Local Agent and Adaptive Edge Privacy Budgeting}

The local agent manages user preferences and budget accounting. For a query $q_i$, it computes a proposed privacy budget
\[
\Delta \varepsilon_i = \big(\varepsilon_{\text{base}} \cdot S_{q_i} \cdot (1 - C_{\text{local}})\big) \cdot F_{\text{seq}} + \eta,
\]
where $\varepsilon_{\text{base}}$ is set from user configuration, $S_{q_i}$ is a sensitivity score in $[0,1]$, $C_{\text{local}}$ is the confidence of local fulfillment, $F_{\text{seq}}$ is a decaying factor based on session length, and $\eta$ is Laplace noise added to protect the budgeting process itself. The effective budget for the query is capped by the user’s remaining allocation, ensuring $\sum_i \varepsilon_i \leq \varepsilon_{\max}$. If local confidence is high, the query is answered locally without consuming budget. Otherwise the delegated query is assigned $\varepsilon_i$ and passed to the translation agent.

\subsection{Translation Agent and Entity Level Differential Privacy}

The translation agent transforms queries according to the allocated budget. We instantiate the mechanism with entity level randomized response. Sensitive entities such as names, identifiers, or dates are detected using simple patterns or named entity recognition. For each detected entity token $t$ from a vocabulary of size $k$, the randomized response mechanism outputs
\[
\Pr[\tilde t = t] = \frac{e^{\varepsilon_i}}{e^{\varepsilon_i} + k - 1}, \quad
\Pr[\tilde t = t'] = \frac{1}{e^{\varepsilon_i} + k - 1} \;\; (t' \neq t).
\]
This satisfies $(\varepsilon_i,0)$ differential privacy per entity under entity adjacency. The probability above provides a tight ceiling on the best possible recovery accuracy, establishing a utility frontier that cannot be exceeded by any estimator. Multiple entities in a query compose sequentially, and multiple queries compose over time under standard DP accounting. The translation agent records transformation parameters and a cryptographic hash of the parameters to produce a proof of transformation that can be checked by the local agent. The proof does not include functions of the original text, ensuring that it is pure post processing and does not degrade privacy.

\subsection{Delegation Gating and Leakage}

Each query may be handled locally or delegated. If the delegation bit $G$ is released, it is privatized with binary randomized response at parameter $\varepsilon_{\text{gate}}$. For any adversary, the Bayes advantage in distinguishing the true delegation decision is bounded by $\tanh(\varepsilon_{\text{gate}}/2)$. If $G$ were released deterministically based on input features, no finite $\varepsilon$ could satisfy differential privacy under entity adjacency, demonstrating the necessity of randomized gating. This impossibility result motivates treating the gating channel explicitly as part of the privacy budget.

\subsection{Remote Agent and Post Processing}

The remote agent is a standard large capacity LLM that receives only transformed queries. Since differential privacy is closed under post processing, the remote model’s output and the verifiability proof add no additional privacy cost beyond the transformation. This separation allows AVEC to remain mechanism agnostic: stronger transformation modules such as embedding perturbation or zero knowledge proofs can replace entity level randomized response without altering the architecture.

\subsection{Privacy Accounting Framework}

To analyze end to end privacy loss, we adopt R\'enyi differential privacy (RDP) as the accounting method. Each delegated step contributes an RDP cost $\varepsilon^{\text{RDP}}_i(\alpha)$ for order $\alpha > 1$, including terms for both the gating mechanism and entity randomized response. These costs add linearly even when budgets are chosen adaptively, consistent with the privacy odometer framework. At the end of a session the total RDP loss is converted to $(\varepsilon,\delta)$ differential privacy using Mironov’s conversion. Because only a fraction of queries are delegated, privacy amplification by subsampling yields tighter bounds for the transformation cost. This formalism provides a general accounting method for AVEC, independent of the specific instantiation of the translation agent.

\section{Experimental Setup}

To evaluate AVEC we implemented a simulation framework that captures the interaction flow between local, translation, and remote agents. The goal of the experiments is to quantify delegation behavior, budget consumption, and cost latency trade offs under controlled conditions while enabling reproducible statistical analysis. All language model functions such as sensitivity assessment and local confidence are simulated with deterministic or randomized heuristics so that trials are repeatable.

Each trial involves one hundred simulated users, each with individualized configuration parameters including privacy preference and maximum cumulative budget. Every user submits ten queries drawn randomly from a pool covering twenty five domains such as medical, financial, legal, and creative contexts. The dataset was constructed to ensure that each query contains elements requiring privacy treatment, for example named entities or identifiers. The order of queries is randomized per user per trial to avoid sequence bias. The complete evaluation consists of thirty independent trials, yielding thirty thousand queries per scenario and enabling robust significance testing across repetitions.

We compare AVEC in both high and medium privacy configurations against several baseline policies. The always delegate baseline sends all queries remotely without transformation. The fixed epsilon baselines use static privacy budgets of 0.1, 1.0, and 5.0 applied to every delegated query. The always local baseline attempts to answer every query on device without delegation, reflecting an extreme low utility policy. These baselines provide reference points for measuring the impact of adaptive budgeting and verifiable transformation.

For each processed query we record delegation status, privacy budget consumed, simulated cost and latency, response quality type, and verification outcome. Costs and latencies are simulated with additive and multiplicative factors tied to transformation complexity and remote inference time. Response quality is categorized by level of privatization, ranging from highly privatized to lightly privatized or local high confidence responses. Verification status indicates whether the proof of transformation generated by the translation agent matched the recomputed proof on the local device. Statistical analysis is conducted on average costs, delegation rates, and distributional differences across scenarios using standard parametric and non parametric tests. We also compute effect sizes to quantify the magnitude of improvements provided by AVEC relative to the baselines.

\section{Formal Proofs}
\label{proofs}

\subsection{Preliminaries}

\begin{definition}[Neighboring datasets]
	Two datasets $D, D'$ are neighbors if they differ in the data of a single individual. Adjacency may be defined at the entity level, where $D$ and $D'$ differ in a single token, or at the user level, where they differ in the entire contribution of one user.
\end{definition}

\begin{definition}[$(\varepsilon,\delta)$-differential privacy]
	A randomized mechanism $M: \mathcal{D} \rightarrow \mathcal{R}$ satisfies $(\varepsilon,\delta)$-differential privacy if for all neighboring datasets $D, D'$ and all measurable subsets $S \subseteq \mathcal{R}$,
	\[
	\Pr[M(D) \in S] \leq e^{\varepsilon} \Pr[M(D') \in S] + \delta.
	\]
\end{definition}

\begin{definition}[Rényi differential privacy \cite{mironov2019rdp}]
	A randomized mechanism $M$ satisfies $(\alpha,\varepsilon)$-RDP of order $\alpha > 1$ if for all neighboring datasets $D, D'$,
	\[
	D_{\alpha}(M(D) \| M(D')) \leq \varepsilon,
	\]
	where $D_{\alpha}$ denotes the Rényi divergence of order $\alpha$. RDP composes additively over multiple mechanisms and can be converted to $(\varepsilon,\delta)$-DP.
\end{definition}

\subsection{Privacy of the AEPB Mechanism}

\begin{lemma}[AEPB privacy]
	Let $f_{\mathrm{AEPB}}$ be the deterministic function that maps query features to a proposed privacy budget $\Delta \varepsilon$. Suppose $\Delta \varepsilon$ is released with Laplace noise $\eta \sim \mathrm{Laplace}(0, b)$. Then the mechanism releasing $\Delta \varepsilon + \eta$ satisfies $(\varepsilon_{\eta},0)$-differential privacy with $\varepsilon_{\eta} = \Delta f / b$, where $\Delta f$ is the $\ell_1$ sensitivity of $f_{\mathrm{AEPB}}$.
\end{lemma}

\begin{proof}
	By the Laplace mechanism, for any two neighboring datasets $D, D'$ and any output $z$,
	\[
	\frac{\Pr[M(D)=z]}{\Pr[M(D')=z]} \leq \exp\!\left(\frac{\Delta f}{b}\right).
	\]
	Thus the release satisfies $(\varepsilon_{\eta},0)$-DP.
\end{proof}

\subsection{Utility Ceiling for Entity Randomized Response}

\begin{theorem}[Entity-DP utility ceiling]
	Consider a vocabulary of size $k$ and the $k$-ary randomized response mechanism with parameter $\varepsilon$. For any estimator $\hat{t}$ attempting to recover the original token $t$ from the privatized output $\tilde{t}$,
	\[
	\Pr[\hat{t} = t] \leq \frac{e^{\varepsilon}}{e^{\varepsilon} + k - 1}.
	\]
\end{theorem}

\begin{proof}
	The Bayes optimal estimator selects the token with the highest posterior probability given $\tilde{t}$. By symmetry of randomized response, this is the true token with probability $\tfrac{e^{\varepsilon}}{e^{\varepsilon} + k - 1}$. No estimator can exceed this bound.
\end{proof}

\subsection{Delegation Gating Privacy}

\begin{theorem}[Delegation leakage bound]
	If the delegation bit $G$ is released using binary randomized response with parameter $\varepsilon_{\mathrm{gate}}$, the maximum Bayes advantage of an adversary distinguishing the true delegation decision is bounded by
	\[
	\text{Adv} \leq \tanh\!\left(\frac{\varepsilon_{\mathrm{gate}}}{2}\right).
	\]
\end{theorem}

\begin{proof}
	For $(\varepsilon,0)$-DP the likelihood ratio between any two outcomes is bounded by $e^{\varepsilon}$. The binary hypothesis testing interpretation of DP implies adversarial advantage bounded by $\tfrac{e^{\varepsilon} - 1}{e^{\varepsilon} + 1} = \tanh(\varepsilon/2)$.
\end{proof}

\begin{theorem}[Impossibility of deterministic delegation]
	If the delegation decision $G$ is a deterministic function of the input query and is not constant across some pair of entity-adjacent queries, then no finite $\varepsilon$ and $\delta < 1/2$ exist such that releasing $G$ satisfies $(\varepsilon,\delta)$-differential privacy.
\end{theorem}

\begin{proof}
	Suppose there exist neighboring queries $q,q'$ such that $G(q) \neq G(q')$. Then $\Pr[G(q)=1]=1$ and $\Pr[G(q')=1]=0$, which violates the DP inequality for any finite $\varepsilon$ with $\delta < 1/2$. Thus deterministic delegation cannot satisfy differential privacy.
\end{proof}

\subsection{Impossibility of Hash-only Proofs}

\begin{theorem}[Impossibility of hash-only DP certification]
	Suppose a verifier observes only a hash of transformation parameters and public policy text, with no attested execution trace. Then there exists no general method for such a verifier to certify that the transformation satisfies $(\varepsilon,\delta)$-DP for all inputs under entity adjacency.
\end{theorem}

\begin{proof}
	Differential privacy is a relational property that requires comparing outputs on pairs of neighboring inputs. A hash of parameters or policy does not bind the actual mapping from input to output. Construct two translation agents that share the same parameter hash but differ on hidden input behavior, one satisfying DP and the other not. The verifier cannot distinguish them. Hence hash-only proofs provide integrity but not privacy certification.
\end{proof}

\subsection{Composition of AVEC Interactions}

\begin{theorem}[End-to-end privacy of AVEC]
	\label{thm:avec-composition}
	Consider a user interacting with the system for $k$ queries. For each query $i$, the local agent privatizes the delegation bit with randomized response at parameter $\varepsilon_{\mathrm{gate}}$ and privatizes each sensitive entity with randomized response at parameter $\varepsilon_i$ chosen adaptively. Let $\mathcal{M}_i$ denote the combined mechanism for query $i$. Then the transcript $\mathcal{T} = (\mathcal{M}_1, \ldots, \mathcal{M}_k)$ satisfies $(\varepsilon_{\mathrm{tot}}, \delta_{\mathrm{tot}})$-differential privacy, where
	\[
	\varepsilon_{\mathrm{tot}} = \mathrm{ConvertRDP}\!\left(\sum_{i=1}^{k} \varepsilon^{\mathrm{RDP}}_{\mathcal{M}_i}(\alpha), \, \delta^{\ast}\right),
	\quad \delta_{\mathrm{tot}} = k \cdot \delta_{\mathrm{ent}} + \delta^{\ast}.
	\]
\end{theorem}
\begin{proof}[Proof sketch]
	Post processing invariance ensures that remote outputs and proofs add no additional cost. Each delegated query contributes RDP cost from gating and entity privatization. Costs add linearly under adaptive choices by the privacy odometer framework \cite{rogers2016privacy}. Advanced composition and conversion to $(\varepsilon,\delta)$-DP yield the bound.
\end{proof}

\section{Threat Model}
\label{threatmodel}

\subsection{System Assumptions}

We assume that the user device is trusted to execute the local agent faithfully and to store configuration data securely. Local execution is not subject to compromise by malware or side channels. Basic cryptographic primitives such as SHA-256 are assumed secure. Communication between agents is encrypted and authenticated, but external observers may monitor traffic patterns. The remote agent is treated as an untrusted black box for privacy but trusted to return correct inference outputs once inputs are received. The translation agent is semi-trusted: it may apply correct transformations or deviate, but it must emit proofs that the local agent can verify.

\subsection{Adversary Classes}

\paragraph{External observers.} Network adversaries can inspect communication between the device and the cloud. Against such adversaries AVEC guarantees differential privacy for delegated queries and randomized protection for the delegation bit. Because budget allocation and transformation occur locally, raw queries never leave the device. Proofs of transformation consist only of declared parameters and hashes, which add no leakage beyond the transformed text.

\paragraph{Semi-honest service providers.} The remote agent is assumed to execute inference correctly but may attempt to learn from the queries it receives. AVEC protects against this by sending only privatized queries, ensuring privacy through differential privacy’s closure under post processing. Thus the remote agent gains no advantage from additional computation beyond the declared privacy bounds.

\paragraph{Malicious translation agents.} Stronger adversaries may attempt to subvert privacy by applying weaker transformations while still emitting syntactically correct proofs. AVEC defends against parameter tampering by requiring local verification of proofs, but the current mechanism ensures only consistency, not full execution integrity. This residual risk motivates future integration of verifiable computation or lightweight zero knowledge proofs, which can bind the execution trace to the declared parameters.

\subsection{Out of Scope}

The current framework does not defend against local device compromise, side channel leakage, or denial of service against the translation agent. It also does not address memorization or membership inference risks inside the remote agent’s internal model weights. These issues are orthogonal to AVEC and require independent mitigations. Within its defined boundaries, AVEC provides differential privacy against honest but curious observers and semi-honest providers, and proof-based integrity checks against malicious transformation behavior.


\section{Findings from Simulation}

The experiments indicates that AVEC substantially reduces delegation and improves cost and latency profiles compared with fixed or static baselines. Across all thirty trials, the dynamic high privacy configuration delegated on average two thirds of queries, compared with one hundred percent delegation under the always delegate baseline. This reduction translates directly into lower cumulative cost and latency, as many queries are answered locally without incurring transformation or remote inference overhead. The effect is consistent across all simulated domains and user profiles.

The statistical analysis suggests that these improvements are highly significant. Analysis of variance across average delegated costs yields an $F$ statistic exceeding 3500 with $p < 0.0001$, rejecting the null hypothesis that all policies perform equivalently. Non parametric verification with the Kruskal–Wallis test leads to the same conclusion. Pairwise comparisons reinforce these findings: the independent samples $t$ test between AVEC high privacy and always delegate produces $p = 0.0001$ with a Cohen’s $d$ greater than 1.1, indicating a very large effect size. The Wilcoxon signed rank test indicates significant cost differences between AVEC and fixed epsilon policies. Delegation rate distributions also differ sharply: the chi squared test between AVEC high privacy and always delegate gives a test statistic over 1000 with $p < 0.0001$ and Cramer’s $V$ around 0.2, signifying a meaningful medium strength association. The Mann–Whitney test on delegation rates corroborates these results.

Beyond statistical significance, the magnitude of the gains highlights the advantage of adaptive budgeting. For delegated queries, the dynamic allocation of $\varepsilon$ allows the system to tailor the level of privatization to query sensitivity, maintaining utility without exhausting budget. The verifiability mechanism ndicates that proofs of transformation match declared parameters, with all simulated trials yielding successful verification and no acceptance of invalid transformations. Together these outcomes show that AVEC can provide not only statistically better performance but also practically relevant reductions in delegation and resource consumption, while ensuring explicit verifiability of privacy enforcement.

\section{Discussion and Limitations}

The evaluation demonstrates that adaptive edge privacy budgeting and verifiable transformation significantly reduce delegation and improve efficiency, but several limitations remain. The current results are derived from a controlled simulation environment rather than deployment with real language models. While simulation allows reproducibility and rigorous statistical testing, it cannot capture the full complexity of model behavior, network variability, or adversarial manipulation. Future work must test AVEC under real conditions with large language model APIs to assess robustness in practice.

A second limitation lies in the heuristic implementation of sensitivity scoring and local confidence. These functions approximate model behavior but are not grounded in actual model predictions. Their effectiveness in guiding budget allocation may vary when applied to diverse real world models, and they could themselves become targets of adversarial exploitation. Similarly, the cryptographic proof of transformation is implemented as a hash based integrity check, which ensures consistency but does not constitute a full privacy proof. Without zero knowledge or verifiable computation, a malicious translation agent could in principle violate privacy while still producing a syntactically correct proof. This motivates the integration of stronger cryptographic methods in future extensions.

The utility analysis is based on categorical response types rather than task specific or human evaluated metrics. While suitable for establishing privacy utility trade offs in simulation, more granular measures are needed for tasks such as summarization or reasoning to fully quantify the impact of transformations. Scalability also remains an open concern: although the local agent naturally scales across devices, the translation agent could face significant overhead in large enterprise deployments where millions of users demand verifiable proofs. Designing distributed or lightweight proof generation is therefore a priority for operational feasibility.

Finally, the framework assumes relative stability of data sensitivity and transformation efficacy. In practice, language models evolve rapidly, and shifts in data distribution may undermine fixed sensitivity heuristics or transformation policies. Adapting to evolving models and adversarial environments requires mechanisms for continuous learning without compromising privacy guarantees. Despite these limitations, the results establish AVEC as a foundation for combining adaptive budgeting, verifiable transformation, and formal differential privacy accounting in LLM inference, setting the stage for future systems that offer both rigorous privacy and practical usability.

\section{Conclusion and Future Work}

This work introduced AVEC, a framework for adaptive edge privacy budgeting and verifiable transformation in large language model interactions. By shifting budget allocation and verification to the user device, AVEC reduces reliance on remote services while enabling explicit auditing of privacy enforcement. The formalization of entity level randomized response within this framework provides clear privacy guarantees, and the use of R\'enyi differential privacy with privacy odometers establishes a principled method for accounting under adaptive allocation. Theoretical results including utility ceilings for entity recovery, bounds on leakage through delegation, and impossibility of deterministic gating or hash only proofs demonstrate that the design rests on firm foundations. Simulations across thirty trials and twenty five domains show that AVEC reduces delegation, lowers cost and latency, and maintains transparent privacy utility trade offs compared with static baselines.

Future work will extend AVEC in several directions. More expressive utility metrics, including task specific and human evaluated measures, are needed to capture the trade off between privacy and effectiveness in real applications. Stronger cryptographic instantiations such as lightweight zero knowledge proofs or verifiable computation can replace the current hash based mechanism to provide end to end assurance against malicious translation agents. Integration with federated learning offers a pathway to use privacy enhanced interactions as signals for local model improvement, addressing the bootstrap problem in developing capable on device LLMs. Scaling the translation agent to enterprise settings will require distributed proof generation and efficient verification protocols. Finally, deployment with real LLM services will allow measurement of adversarial robustness and user experience in practice. By addressing these challenges, AVEC can evolve from a foundational framework to a practical system for trustworthy private language model deployment.
\bibliography{avec_refs}
\bibliographystyle{iclr2026_conference}

\appendix
\section{Simulation Details}
\label{appendix:simdetails}

This appendix describes the simulation environment used to evaluate AVEC. The framework models the interaction among a local agent, a translation agent, and a remote agent. All components are implemented in Python 3 on a Linux system with numpy, pandas, and scipy, using seeded pseudo random generators to ensure reproducibility. A master seed of 1729 is used, with per trial seeds deterministically derived as 1729 plus the trial identifier.

Each trial involves one hundred users, each assigned a configuration with a privacy preference (high or medium) and a maximum cumulative budget in the range one to two. Every user issues ten queries, giving one thousand queries per trial. Queries are sampled from a pool covering twenty five domains, including medical, financial, legal, creative, environmental, and security contexts. Each query is designed to contain at least one sensitive entity. The order of queries is randomized per user and per trial. Thirty independent trials are executed for each scenario, producing thirty thousand queries and supporting robust statistical analysis.

The local agent applies adaptive budgeting with base values of 0.05 for high privacy and 0.10 for medium privacy. The sequence decay factor is set as $F_{\mathrm{seq}} = \exp(-n/\kappa)$ with $\kappa=5$, and Laplace noise with scale $\Delta f / \varepsilon_{\eta}$, where $\Delta f=0.3$ and $\varepsilon_{\eta}=0.01$, is added to protect the budgeting logic. Local confidence is computed from a heuristic classifier combining cache hits and lightweight features, with a threshold of 0.8 required for local answers. A small cache of common facts is provided to simulate certain high confidence responses. The local agent enforces a hard budget cap to guarantee that the total expenditure never exceeds the configured maximum.

The translation agent implements entity level randomized response. Sensitive entities are detected using regular expressions and named entity recognition, and privatized within a vocabulary of bounded size, typically between eight and sixty four. Transformation parameters, the list of redacted or substituted entities, the effective budget, a policy identifier, and a timestamp are recorded, along with a SHA-256 hash over a canonical serialization of these fields. The proof does not contain any function of the raw input. Latency is modeled as a base delay between five and fifty milliseconds plus a term proportional to input length, and cost is modeled as a base value between $2 \times 10^{-4}$ and $2.5 \times 10^{-3}$ with an additional length dependent increment.

The remote agent models a generic large language model API. Base latency is drawn uniformly between 300 and 800 milliseconds and base cost between 0.007 and 0.015 in normalized units. Multiplicative modifiers reflect the difficulty of answering privatized queries: highly privatized inputs incur factors between 0.4 and 0.6, moderately privatized between 0.7 and 0.9, lightly privatized between 0.9 and 1.1, and unmodified queries between 0.95 and 1.05. Responses are labeled as local high confidence or remote with varying levels of privatization.

We evaluate AVEC in high and medium privacy settings and compare against five baselines: always delegate, fixed epsilon with parameters 0.1, 1.0, and 5.0, and always local. For every query, we log the delegation decision, the budget consumed, the simulated cost and latency, the response label, and the verification result. Verification succeeds if the recomputed SHA-256 hash of the proof fields matches the received proof. Additional logs include the number of privatized entities, the effective vocabulary size, and the remaining budget after each query.

Statistical tests are conducted on per trial averages. Analysis of variance across scenarios is complemented by Kruskal–Wallis tests, while pairwise comparisons use independent samples $t$ tests and Wilcoxon signed rank tests. Delegation distributions are compared using chi squared and Mann–Whitney tests. Effect sizes are reported as Cohen’s $d$ for numerical comparisons and Cramer’s $V$ for categorical comparisons. Significance is assessed at $p < 0.05$ with Bonferroni correction for families of tests, and confidence intervals are computed by nonparametric bootstrap with one thousand resamples. All configuration parameters, seeds, and query pools are versioned, and each scenario produces deterministic aggregate statistics, ensuring reproducibility of all reported results.

\section{Broader Impact}

AVEC addresses a central obstacle in deploying language models responsibly by enabling verifiable privacy guarantees at inference time. By shifting privacy budgeting and enforcement to the edge, the framework gives users meaningful control over how their queries are processed and logged. This shift has the potential to broaden adoption of LLMs in sensitive domains such as healthcare, finance, education, and public services, where data exposure remains a key barrier. The bootstrap problem in local model development can also be mitigated by allowing privacy enhanced telemetry to be reused for on device training, accelerating the emergence of capable local models without compromising privacy. Explicit verifiability further increases user trust by moving from implicit assurances to cryptographically auditable artifacts. Taken together, these elements can foster confidence in responsible AI deployment and create opportunities for models to serve populations that would otherwise be excluded on privacy grounds. At the same time, AVEC contributes to resource efficiency by reducing reliance on repeated remote calls, potentially lowering both operational cost and environmental impact.

\section{Ethics}

Despite its advantages, AVEC raises important ethical considerations. The framework protects user inputs but does not inherently prevent harmful or biased outputs from the remote model. Systems built on AVEC must therefore incorporate independent content moderation and fairness auditing to avoid amplifying misinformation or reinforcing stereotypes. The budgeting mechanism creates a new axis of disparity: users who exhaust their privacy budget quickly may experience degraded service, which raises concerns of unequal access. Transparent communication of budget status and careful choice of default values are necessary to reduce the risk of unfair outcomes. The use of heuristic sensitivity scoring also introduces the possibility of systematic bias if certain user groups have their queries consistently judged more sensitive or less confident, thereby affecting service quality. Accountability mechanisms must extend beyond aborting unverified queries, including logging and audit trails that balance transparency with privacy. Finally, by making privacy enforcement cryptographically verifiable, AVEC may complicate internal system auditing if logs contain only transformed queries. Future deployments should explore privacy preserving telemetry that allows operators to monitor system health without undermining user protections. Careful attention to these ethical challenges is essential for responsible use of the framework.
\section{Extended Related Work}
\label{appendix:relatedwork}

This appendix provides an extended survey of prior work related to AVEC. The main text highlights only the most directly relevant contributions; here we situate our framework more broadly across large language model agents, differential privacy, edge learning, verifiable computation, anonymization, and orchestration.

\subsection{LLMs and Agent Paradigms}
Transformer-based large language models \cite{vaswani2017attention} have enabled the rise of autonomous LLM agents that perform multi-step reasoning and action. Notable examples include ReAct \cite{yao2023react}, which combines reasoning traces with action calls, and Voyager \cite{wang2023voyager}, an open-ended embodied agent. Surveys such as Wang et al. \cite{wang2024surveyagents} summarize the emerging design space of LLM-based agents. These systems typically assume powerful remote models, whereas AVEC addresses the distinct challenge of bootstrapping privacy-preserving local agents.

\subsection{Differential Privacy in Machine Learning}
Differential privacy (DP) provides rigorous guarantees against individual information leakage \cite{dwork2014algorithmic}. Its application to deep learning was established through DP-SGD \cite{abadi2016deep}, while advanced analysis has leveraged R{\'e}nyi DP \cite{mironov2019rdp}. Composition under repeated use has been studied via privacy odometers and filters \cite{rogers2016privacy}. AVEC builds on these foundations by proposing adaptive, edge-based per-query budgeting. 

\subsection{Edge AI and Federated Learning}
Federated learning (FL) \cite{mcmahan2017communication} enables decentralized model training without centralizing data, offering strong privacy benefits. Extensions such as edge-based DP collection \cite{jiang2021privacy} further strengthen privacy by acting at the data source. AVEC aligns with this paradigm but focuses on inference-time controls: adaptive DP allocation and verifiable transformation at the edge.

\subsection{Verifiable Computation and Cryptographic Proofs}
Verifiable computation (VC) and zero-knowledge proofs (ZKPs) aim to establish correctness of computation without revealing sensitive inputs. Bulletproofs \cite{boneh2018bulletproofs} and earlier work on non-interactive VC \cite{gennaro2010non} represent important milestones. Surveys such as Peng \cite{peng2025survey} outline recent progress on proofs for machine learning. AVEC departs from full cryptographic proofs in favor of lightweight hash-based proofs of transformation, leaving stronger instantiations to future work.

\subsection{Privacy-Preserving Anonymization and Text Analysis}
Classical data anonymization frameworks include $k$-anonymity \cite{sweeney2002kanon}, $\ell$-diversity \cite{machanavajjhala2007ldiversity}, and $t$-closeness \cite{li2006t}, each offering varying robustness against re-identification. Surveys on text privacy \cite{majeed2020anonymization} highlight challenges unique to language data. More recent work explores semantic generalization for privacy-preserving release \cite{iwendi2020n}. AVEC connects to this literature through its use of entity-level randomized response as a text transformation mechanism.

\subsection{LLM Orchestration and Caching}
Agent orchestration frameworks for LLMs \cite{zou2025survey} and semantic caching approaches \cite{li2024scalm} aim to improve efficiency and coherence of multi-step interactions. On-device language models \cite{xu2024device} are increasingly considered to mitigate latency and privacy costs. AVEC complements these efforts by embedding adaptive privacy budgeting and verifiability into the orchestration loop.

\subsection{Adversarial Robustness Context}
Finally, adversarial robustness studies \cite{goodfellow2014explaining} illustrate how small perturbations can manipulate model behavior. While not the prima

\end{document}